\documentclass[12pt,onecolumn]{IEEEtran}
\makeatletter
\def\ps@headings{%
\def\@oddhead{\mbox{}\scriptsize\rightmark \hfil \thepage}%
\def\@evenhead{\scriptsize\thepage \hfil \leftmark\mbox{}}%
\def\@oddfoot{}%
\def\@evenfoot{}}
\makeatother
\usepackage{setspace}
\usepackage{geometry}
\usepackage{graphicx}
\usepackage{amsmath,amssymb}
\usepackage{amsmath}
\usepackage{multirow}
\usepackage{verbatim}
\usepackage{epstopdf}
\usepackage{epsfig}
\usepackage{cite}
\usepackage{mdframed}
\usepackage{subfigure}
\usepackage{graphics,color}
\usepackage{lipsum}

\newcommand{\qed}{\rule{2mm}{3mm}}
\newtheorem{theorem}{Theorem}[section]
\newtheorem{remark}[theorem]{Remark}
\newtheorem{proposition}[theorem]{Proposition}
\newtheorem{lemma}[theorem]{Lemma}
\newtheorem{corollary}[theorem]{Corollary}
\newtheorem{assumption}{Assumption}

\newtheorem{proof}{Proof}

\let\OLDthebibliography\thebibliography
\renewcommand\thebibliography[1]{
  \OLDthebibliography{#1}
  \setlength{\parskip}{0pt}
  \setlength{\itemsep}{0pt plus 0.3ex}
}

\ifodd 0
\newcommand{\rev}[1]{{\color{blue}#1}} 
\newcommand{\com}[1]{\textbf{\color{red}(COMMENT: #1)}} 
\newcommand{\clar}[1]{\textbf{\color{green}(NEED CLARIFICATION: #1)}}
\newcommand{\response}[1]{\textbf{\color{magenta}(RESPONSE: #1)}} 

\else
\newcommand{\rev}[1]{#1}
\newcommand{\com}[1]{}
\newcommand{\clar}[1]{}
\newcommand{\response}[1]{}
\fi

\begin{document}
\title{To Stay Or To Switch: Multiuser Dynamic Channel Access 
\thanks{Y. Liu and M. Liu are with the Electrical Engineering and Computer Science Department,  University of Michigan, Ann Arbor, \{youngliu, mingyan\}@umich.edu. A preliminary version of this paper appeared in INFOCOM in April 2013.  This work was partially supported by the NSF under grants CIF-0910765 and CNS-1217689, and the ARO under Grant W911NF-11-1-0532.}
}
\author{
\IEEEauthorblockN{Yang Liu and Mingyan Liu}\\
}

\maketitle
\pagestyle{plain}
\begin{abstract}
In this paper we study opportunistic spectrum access (OSA) policies in a multiuser multichannel random access cognitive radio network, where users perform channel probing and switching in order to obtain better channel condition or higher instantaneous transmission quality.  However, unlikely many prior works in this area, including those on channel probing and switching policies for a single user to exploit spectral diversity, and on probing and access policies for multiple users over a single channel to exploit temporal and multiuser diversity, in this study we consider the collective switching of multiple users over multiple channels.  In addition, we consider finite arrivals, i.e., users are not assumed to always have data to send and demand for channel follow a certain arrival process.  Under such a scenario, the users' ability to opportunistically exploit temporal diversity (the temporal variation in channel quality over a single channel) and spectral diversity (quality variation across multiple channels at a given time) is greatly affected by the level of congestion in the system. We investigate the optimal decision process in this case, and evaluate the extent to which congestion affects potential gains from opportunistic dynamic channel switching.

\end{abstract}
\vspace{1em}

\section{Introduction}\label{sec:intro}

Dynamic and Opportunistic Spectrum Access (OSA) policies have been very extensively studied in the past few years for cognitive radio networks, against the backdrop of spectrum open access as well as advances in ever more agile radio transceivers, including e.g., highly efficient channel sensing techniques \cite{5454399,4336106}.  Within this context, a cognitive radio is capable of quickly detecting spectrum quality and performing channel switching so as to obtain good channel and transmission quality.  At the heart of such opportunistic spectrum access is the idea of improving spectrum efficiency through the exploitation of {\em diversity}.

Within this context there are three types of diversity gains commonly explored.  The first is {\em temporal diversity}, where the natural temporal variation in the wireless channel causes a user to experience or perceive different transmission conditions over time even when it stays on the same channel, and the idea is to have the user access the channel for data transmission when the condition is good, which may require and warrant a certain amount of waiting.  Studies like \cite{DBLP:journals/ton/ChangL09} investigate the tradeoff involved in waiting for a better condition and when is the best time to stop.

The second is {\em spectral diversity}, where different channels experience different temporal variations, so for a given user at any given time a set of channels present different transmission conditions.  The idea is then to have the user select a channel with the best condition at any given time for data transmission, which typically involves probing multiple channels to find out their conditions.  Protocols like \cite{Kanodia_moar:a} does exactly this, and studies like \cite{DBLP:journals/jsac/ZhaoTSC07,Ahmad:2009:OMS:1669634.1669645} further seek to identify the best sequential probing policies using a decision framework.

The third is {\em user diversity} or {\em spatial diversity}, where the same frequency band at the same time can offer different transmission qualities to different users due to their difference in transceiver design, geographic location, etc. The idea is to have the user with the best condition on a channel use it.  This diversity gain can be obtained to some degree by using techniques like stopping time rules whereby a user essentially judges for itself whether the condition is sufficiently good before transmitting, which comes as a byproduct of utilizing temporal diversity.

We note that the above forms of diversities are often studied in isolation.  For instance, temporal diversity is studied in a multiuser setting but with a single channel in \cite{Zheng:2009:DOS:1669334.1669354,Tan:2010:DOS:1833515.1833882}; spectral diversity is analyzed for a single user in \cite{Shu:2009:TSC:1614320.1614325}, among others.
More specifically, \cite{Zheng:2009:DOS:1669334.1669354} used temporal diversity in a multi-user setting and developed optimal stopping policies \cite{uclastopping}. \cite{Tan:2010:DOS:1833515.1833882} considered a distributed opportunistic scheduling problem for ad-hoc communications under delay constraints.
In \cite{Shu:2009:TSC:1614320.1614325} authors exploited spectral diversity in OSA for a single user with sensing errors,  
where the multi-channel overhead is captured by a generic penalty on each channel switching.  This becomes insufficient in a multi-user setting as such overhead will obviously depend on the level of congestion in the system that results in different amount of collision and the time it takes to regain access to a channel.
%
In \cite{Kanodia_moar:a} an opportunistic auto rate multi-channel MAC protocol MOAR is presented to exploit spectral diversity for a multi-channel multi-rate IEEE 802.11-enabled wireless ad hoc network.  However, this scheme does not allow parallel use of multiple channels by different users due to its reservation mechanism.
Other works that study multi-channel access by a single user include \cite{DBLP:journals/ton/ChangL09,4723352,Ahmad:2009:OMS:1669634.1669645,Liu_indexabilityof,citeulike:6090701,4604743}.

As the number of users and their traffic volume increase in such a multi-channel system, one would expect their ability to exploit the above diversity gains to decrease significantly due to the increased overhead, e.g., the time it takes to perform channel sensing or the time it takes to regain access right, or increased collision due to channel switching.
As mentioned above, this overhead was captured in the form of penalty cost in prior work such as \cite{Shu:2009:TSC:1614320.1614325}, but is often assumed to be independent of the traffic volume existing in the system.

With the above in mind, in this paper we set out to study opportunistic spectrum access policies in a multiuser multichannel random access setting, where users are not assumed to always have data to send, 
demand for channel follows a certain arrival process, and collision and competition times are taken into account. Our focus is on the effect of collective switching decisions by the users, and how their decision process, in particular their channel switching decisions, are affected by increasing congestion levels in the system.
%

Toward this end we characterize the nature of an optimal access policy and identify conditions under which channel switching actually results in transmission gain (e.g. in terms of average data rate or throughput).  Our qualitative conclusion, not surprisingly, is that with the increase in user/data arrival rate, the average throughput decreases and a user becomes increasingly more reluctant to give up a present transmission opportunity in hoping for better condition later on or in a different channel.  Quantitatively we present algorithms that calculate optimal switching decisions and analyze the stability of the overall system.

The remainder of this paper is organized as follows.
The system model is given in Section \ref{model}. In Sections \ref{iid} and \ref{sec:markov}, we model channel evolution as IID and Markovian processes, respectively, and analyze the properties of an optimal stopping/switching rule.  Numerical results are given in Section \ref{simu}, and Section \ref{conclusion} concludes the paper.

\section{Model, assumptions and preliminaries}\label{model}

\subsection{Model and assumptions}

Consider a wireless system with $N$ channels indexed by the set $\Omega = \{1,2,...,N\}$. We associate each channel with a positive reward of transmission (e.g., transmission rate) $X^j$, which is a positive random variable with distribution given by $f_{X^j}(x)$, assumed to have finite support with a maximum value of $\bar{X}^j$. There are $m$ cognitive users (or radio transceivers) each equipped with a single transmitter attempting to send data to a base station.  Our model also captures direct peer-to-peer communication, where $m$ {\em pairs} of users communicate and each pair can rendezvous and perform channel sensing and switching together through the use of a control channel \cite{Liu12}.  However, for simplicity of exposition, for the rest of the paper we will take the view of $m$ users transmitting to a base station.  We will assume these $m$ users are within a single interference domain, so that at any given time each channel can only be occupied by one user.  Considering spatial reuse will make the problem considerably more challenging and remains an interesting direction of future research.

We consider discrete time with a suitably chosen time unit, and with all other time values integer multiples of this underlying (and possibly very small) unit.  We will consider two channel models, an IID model where channel conditions over time are assumed to form an IID process defined on this time unit to model fast changing channels (in Section \ref{iid}), and a Markovian model where channel conditions over time form a Markov chain for modeling slow changing channels (in Section \ref{sec:markov}).   Different channels are in general {\em not} identically distributed, and evolved independently of each other.  

\begin{figure}[!h]
\centering
                \includegraphics[width=0.7\textwidth,height=0.32\textwidth]{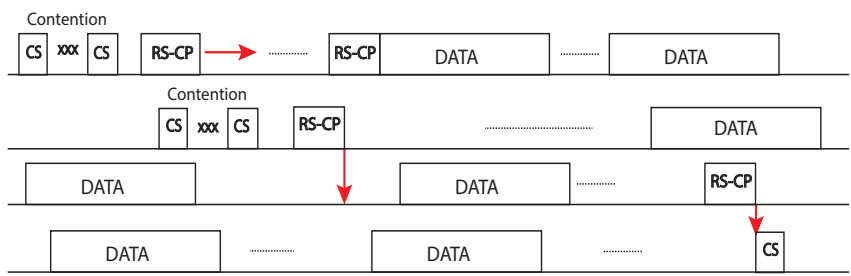}\label{sysmodel}
\caption{System model. 
}
\end{figure}

The system operates in a way similar to a multi-channel random access network like IEEE 802.11, with the following modifications related to dynamic and opportunistic channel access.  Each user has a pre-assigned (or self-generated) random sequence of channels; this sequence determines in which order the user performs channel switching, an approach similar to that used in \cite{Shu:2009:TSC:1614320.1614325}.  More on this assumption is discussed in Section \ref{sec:formulation}.
Each time a user enters a new channel, it must perform carrier sensing (CS) and compete for access (contention resolution) as in a regular 802.11 channel.  As soon as it gains the right to transmit, \rev{the user reserves the channel (e.g., through the use of RTS-CTS type of handshake) and finds out the instantaneous data transmission quality (channel information could be piggybacked on these control packets) it may get if it transmits right away}.  Upon finding out the channel condition, this user faces the following choices:
\begin{enumerate}
\item Transmit on the current channel right away.  Intuitively this happens if the current channel condition is deemed good enough.  This action will be referred to as {\em STOP}.  \rev{This is shown in Fig. \ref{sysmodel}, where the second RS-CP (denoting the Reservation-Channel probing process) followed by DATA indicates a STOP at the first channel (first line in the figure).} 

\item Forego this transmission opportunity, presumably due to poor channel condition, but remain on the same channel and compete for access again in the near future hoping to come across a better condition then.  This happens if the current channel condition is poor but the average quality is believed to be good, so the user will risk waiting for possibly better condition later.  This action will be referred to as {\em STAY}. \rev{This is illustrated by the first RS-CP on the first line (channel) followed by a horizontal red arrow. } 

\item Give up the current channel and switch to the next one on its list/sequence of channels.  This happens if the current channel condition is poor, and the prospect of getting better conditions later by staying on the same channel is not as good as by switching to the next channel.  This action will be referred to as {\em SWITCH}. \rev{An example is shown by the RS-CP on the second line (channel) followed by a vertical red arrow indicating a {\em SWITCH} action. }
\end{enumerate}
Note that option \rev{(2)} above allows the system to exploit both multiuser diversity (the transmission opportunity is given to another user under the random access) and temporal diversity (the user in question waits for better condition to appear in time), while option (3) allows the system to exploit spectral diversity as users seek better conditions on other channels.
These options, in particular (1) and (2) are similar to those used in a stopping time framework, see  e.g., \cite{Zheng:2009:DOS:1669334.1669354}.

In the above decision process once a user decides to leave a channel it cannot use the channel for transmission without going through carrier sensing and random access competition again.  More importantly from a technical point of view, this assumption means that the user cannot claim the same channel condition at a later time.  
Once a user gets the right to transmit on a certain channel, it can transmit for a period of $T$ time units, which is a constant.  For simplicity a single time unit is assumed to be the amount of time to transmit a control packet (\rev{e.g., RTS/CTS type of packets.}). 
\subsection{Capturing the level of congestion}

As mentioned earlier our focus in this paper is on understanding how the users' channel access
decision process is affected by increasing traffic load or congestion in the system. To model this we will first take the view of a single user, and introduce user arrival rates in each channel as well as the amount of delay involved in STAY and SWITCH as parameters that need to be taken into consideration in its decision process.  Note that these parameter values are the result of the collective switching actions of all users, and therefore cannot be obtained prior to defining the switching policies.  We will however assume that these parameters have well-defined averages to facilitate our analysis.  Later on we show that the system under the optimal switching policy converges and that these parameters indeed have well-defined averages, thereby justifying such an assumption.  In other words, policies derived under the assumption that these parameters have well-defined averages lead to a stable system with well-defined averages for these parameters.  This is not unlike the Markov mean-field approach where a single user operates against a background formed by all other users in system over which this single user has no control or influence.   In practice these values may be obtained by users through learning.



Specifically, we assume that the total packet arrivals to a channel, including external arrival, retransmission, as well as arrivals switched from other channels, form a Poisson process, with the attempt rate vector given by
$\mathbf{G} = [G_1,G_2,...,G_N]$ and a sum rate $\sum_{i=1}^N G_i = G$. These quantities will also be referred to as the {\em load} or {\em traffic load} on a channel.  We will not directly deal with the external arrival processes as our analysis entirely depends on the above ``internal'' offered load.  However, we will assume that the external arrivals are such that the system remains stable. 


The level of congestion on any channel is captured by two parameters. 
The first is an average {\em contention delay} on channel $j$ denoted by $t_j^c$; this is the average time from carrier sense to gaining the right to transmit on channel $j$.  The more competing users there are on channel $j$, the higher this quantity is.  
The second is an average {\em switching delay} of channel $j$, denoted by $t_j^s$; this is the time from a user switching into channel $j$ (from another channel) to its gaining the right to transmit on channel $j$. Compared to $t_j^c$ the switching delay includes the additional time for the radio to perform channel switching and additional waiting time in the event that the switching occurs during an active transmission.  \rev{In our characterization of $t_j^s$ below, however, we will ignore the hardware switching delay as it simply adds a constant, very small  compared to contention delay, which will not affect our subsequent analysis.} 

For a packet arriving at channel $i$ (from an external arrival process or by switching from another channel), the delay it experiences between arrival and successful transmission consists of two parts, the average time it takes for the channel to become idle if it happens to arrive during an active transmission (including its associated control packet exchange), denoted by $t^{w}_i$, and the average time it takes to compete for and gain the right to transmit, given by $t^c_i$. 
We thus have $t^s_i = t^{w}_i + t^c_i$. 


Denote by $Y$ the random variable representing the time between a new arrival and the completion of the current transmission. Following results in \cite{Liu12}, we have
$f_Y(y) = S_{i} e^{-S_{i} y}$, where $S^{i}$ is the success rate of channel contention given by
\begin{align}
S_{i} = \frac{G_{i}e^{-2G_{i}}}{1+(1+T)G_{i}e^{-2G_{i}}}~.
\end{align}
$t^{w}_i$ is then calculated as follows:
\begin{align}
&t^{w}_i = \int_{0}^{1+T}f_Y(y)(1/\zeta+y) dy = \frac{1}{S_{i}}+\frac{1}{\zeta}-(T+1+\frac{1}{S_{i}}+\frac{1}{\zeta})e^{-(T+1)S_{i}}~,
\end{align}
where $1/\zeta$ is the expected random backoff time.  For $t^c_i$, since a competition succeeds with probability $e^{-2G_{i}}$ we have
\begin{align}
t^c_i = (e^{2G_{i}}-1)\cdot(1/\zeta+2)+2~.
\end{align}

Using the above expressions, it is not difficult to establish the following results.
\begin{proposition}
\label{assump-t1}
Both $t_{j}^c$ and $t_{j}^s$ are non-decreasing functions of arrival rate $G_j$, $\forall j \in \Omega $. 
\end{proposition}
\begin{proposition}\label{assump-t2}
Both $t^c_{j}$ and $t^s_{j}$ are non-decreasing functions of the data transmission time $T$, $\forall j \in \Omega$.
\end{proposition}

The decision process we introduce next is a function of $t^c_i$ and $t^s_i$, so a user needs to know these parameter values in order to compute the optimal policy.  In practice, this information may be obtained through measurement and empirical means. 


%

\subsection{Problem formulation}\label{sec:formulation}

For simplicity and without loss of generality, for the single user under consideration we will relabel the channels in its sequence in the ascending order: $1, 2, \cdots, N$.
We now define the following rate-of-return problem with the objective of maximizing the effective data rate over one successful data transmission.

Specifically, let $\pi$ denote a policy $\pi=\{ \alpha_1, \alpha_2, \cdots \alpha_{\gamma(\pi)}\}$ which specifies the sequence of actions leading up to a successful transmission, with $\alpha_k$ denoting the $k$-th action, $\alpha_k \in \{\mbox{STAY, SWITCH}\}$, $k=1, \cdots, \gamma(\pi)-1$, and $\alpha_{\gamma(\pi)}=\mbox{STOP}$.  An action is only taken upon gaining the right to transmit in a channel, and $\gamma(\pi)$ denotes the stopping time at which the process terminates with a transmission action.

Let $X_{\gamma(\pi)}^\pi$ denote the data rate obtained at the last step when the process terminates.  Then the total reward the user gets is $X_{\gamma(\pi)}^\pi \cdot T$, the total amount of data transmitted.  A natural goal would be to maximize the ratio between this reward and the total amount of time spent in the decision process (summing up the delays involved in switching and contention as a result of the actions), i.e., the effective or average throughput or data rate. 
%
While this appears to be a standard rate-of-return problem, an inherent difficulty arises from the fact that different channels have different statistics, and thus the rewards generated and the delays experienced, respectively, are not independent across channels.  This prevents the use of the renewal theorem to turn the expectation of the aforementioned ratio (average throughput) into a ratio of expectations as is commonly done. 

To address this difficulty, we will make the following simplification: instead of maximizing the overall rate of return for each successful transmission over the entire decision process, we will seek to maximize the rate of return over the {\em remaining} decision process given the current state of the process. This may be viewed as a ``no-recall'' approximation to the original goal by ignoring the history or past decisions in the same process.  This objective can be represented by the following dynamic program, noting that the user goes through the channels in the order $1, 2, \cdots, N$. 
\begin{align}
V_{N}(x) &= \max \left\{x, 
\frac{T}{T+t^c_N}E \{V_{N}(X^N|x)\}\right\}\nonumber \\
V_{i}(x) &= \max \left\{x, 
\frac{T}{T+t^c_i}E \{{V}_{i}(X^i)|x\},
\frac{T}{T+t^s_{i+1}}E \{{V}_{i+1}(X^{i+1})|x\}\right\} \label{pf}, ~~~ i<N, 
\end{align}
where $V_{i}(x)$ is the value function at stage $i$ (in channel $i$) of the decision process when the observed channel state is $x$; this is also the maximum average throughput obtainable given current state $x$ (transmission rate) in channel $i$: 
In the above equation, the first term is the reward (current transmission rate) if we STOP, the second the expected reward ( if we STAY, 
and the last the expected reward if we SWITCH. 

The optimal decision process defined by (\ref{pf}) appears to be a finite horizon problem, i.e., the process stops at channel (or stage) $N$.  However, this is only partly true.  (\ref{pf}) actually illustrates a {\em two-dimensional} decision problem, where there is a finite number of steps ($N$) along the spectral dimension (the channels), but within each channel (for each $i$) the decision process is over an infinite horizon along the time dimension, i.e., the decision process may go on indefinitely within a particular channel.  This will be seen more clearly in Section \ref{iid}.

The reason we have limited the horizon to be finite along the spectral dimension -- the infinite horizon version would be where the user can continue to switch channels for an indefinite number of times, including revisiting channels it has visited in the past -- has to do with the IID assumption on the channels.  Since channel state realizations are independent over time (for the same channel), the second and third terms in (\ref{pf}) are both independent of the current state $x$.  In other words, the comparison between the second and the third terms is independent of the current state $x$, suggesting that if the second term is larger than the third term, then it will always be larger regardless of the current state.  The interpretation of this observation is that if we ever decide to STAY (the second term is larger) on the same channel, then we will never SWITCH (the third term is larger) later. The opposite is also true: if we ever decide to SWITCH away from a channel (the third term is larger), then under the optimal policy we will never come back to the same channel even if we are allowed to.  This means that under the objective of maximizing the future rate of return, we will never visit more than once each channel, resulting in the finite horizon along the spectral dimension given above; there is no need to allow the user to revisit a channel it has visited before but switched away from.

The reason why we also limit ourselves to a pre-determined sequence of channels has to do with the multi-user scenario we aim to analyze.  If there is only a single user, then obviously the reasonable thing to do is to also optimize the sequence/order of channel sensing, together with optimizing the switching and transmission decisions.  Indeed there has been a large volume of study on determining optimal sensing orders, see e.g., \cite{4786510,5272267,5738213}.  However, when there are multiple competing users this type of optimization is no longer applicable: one's previously optimal sensing order may no longer be optimal depending on what order other users adopt.  Consequently this needs to be either treated as a centralized multi-user optimization problem, where the jointly optimal sensing orders are computed simultaneously for all users, or treated as a game-theoretic problem where each user selfishly determines its sensing order to maximize its own utility.  Either approach is very different from the study in the present paper.  The game-theoretic approach in particular is largely an open question as it involves the equilibrium analysis of complex decisions (not only the sensing order of channels but also the stopping decisions on any given channel). While this remains an interesting directly of future research, in the present study we adopt the assumption that a user simply follows a pre-defined (can be randomly chosen) sequence of channels and focus our attention on the switching decisions instead.  In Section \ref{simu} we compare the result between randomly selecting these sequences and greedily selecting these sequences where each user always sense channels in descending order of the average reward.

For the remainder of our presentation, we will use the terms {\em stages} and {\em steps} to describe the two time scales of decision making along the two dimensions described above.  Movement along the spectral dimension (i.e., switching from one channel to the next) occur in stages; stage $i$ means channel $i$ and this is indexed by the subscript in the value function $V_i(x)$.  The decision process within the same stage (or in the same channel) occur in steps; the decision to remain on the same channel or switch away occur at the boundary of a step.  The indexing of steps is not explicit in the expression given in (\ref{pf}) but will be made explicit in our subsequent analysis.

\section{Optimal access policy under the IID channel model}\label{iid}
In this section, we model the channels as fast changing, IID processes, where successive observations of the state of the same channel are independent. 

\subsection{An optimal ``nested'' stopping rule}
Since successive channel states are independent, the value function (\ref{pf}) is simplified:
\begin{align}\label{eqn:simple}
V_{i}(x) = \max \left\{x, 
\frac{T}{T+t^c_i}E \{V_i(X^i)\}, \frac{T}{T+t^s_{i+1}}E\{V_{i+1}(X^{i+1})\}\right\} ~.
\end{align}
The above three-way comparison suggests the following. If the current state $x$ is sufficiently high then the optimal decision is STOP.  The comparison between the second and the third terms is more interesting: both terms are independent of the state $x$, so if the second term is larger then it will always be larger.  As previously mentioned, this implies that if we ever decide to STAY, then we will never SWITCH later.  The reverse is also true: if we ever decide to SWITCH then we will never return to the same channel.  These observations can lead to a concrete proof of the existence and uniqueness of a threshold rule but in general cannot produce a closed  form for the computation of the threshold.  Below we will instead use results from optimal stopping theory \cite{uclastopping} to obtain not only the existence but also a closed form for the threshold. Consider the following substitution,
\begin{align}\label{eqn:bundle}
\hat{X}^{i}(x) = \max \left\{x, 
\frac{T}{T+t^s_{i+1}}E\{V_{i+1}(X^{i+1})\}\right\}
\end{align}
with the value function subsequently re-written as
\begin{align}\label{eqn:transform}
V_{i}(x) = \max \left\{\hat{X}^{i}(x) ,\frac{T}{T+t^c_i}E \{V_i(X^i)\}\right\} ~.
\end{align}

This substitution reduces the decision process to a two-way comparison, and more importantly, a one-dimensional decision process. Specifically, since the state $x$ is IID over the same channel/stage $i$, the first term $\hat{X}^{i}(x)$ as defined in (\ref{eqn:bundle}) is also IID over the same stage $i$ while encoding the information of other channels/stages.  Therefore, if we view $\hat{X}^{i}(x)$ as the reward of a (meta) stopping action and $t^c_i$ as the cost for continuing, then the value function given in (\ref{eqn:transform}) represents a standard stopping time rate-of-return problem with two possible actions in each step, (meta) stopping and continuation, respectively, and this process concerns only a single stage/channel. The switching to the next stage occurs when the (meta) stopping action is taken (which essentially ends the above one-dimensional stopping time problem), and it is determined that SWITCH is a better action than STOP. 

The following theorem characterizes the property of the optimal decision for the problem given in (\ref{eqn:simple}) or equivalently (\ref{eqn:transform}).
\begin{theorem}
The optimal action at stage $i$ of deciding between \{STOP, SWITCH\} and STAY is given by a stopping rule: the state space of the channel condition can be divided into a stopping set $\Delta_i^s$ and continuation set $\Delta_i^c$, such that whenever the channel condition is observed to be in either set, the corresponding action (STOP/SWITCH vs. STAY) is taken\footnote{The word ``continuation'' in this context refers to continuing on the same channel, whereas ``stopping'' (or the term (meta) stopping used earlier) refers to no longer staying on the same channel either by a transmission or by switching away.}.
Furthermore, these two sets are given by the following threshold property: 
\begin{align}
\Delta^s_i =  \{x: \hat{X}^{i}(x)  \geq \lambda^*\}, \forall i~,
\end{align}
where the threshold $\lambda^*$ at the $i$th stage is given by the unique solution to
\begin{align}
E[\hat{X}^{i}(x)  - \lambda]^+ = \frac{\lambda \cdot t^c_i}{T}~. \label{sol}
\end{align}
\end{theorem}

\begin{proof}
We first prove the existence of an optimal stopping rule. Define the reward function associated with step $k$ of the stopping decision process at stage $i$ as
\begin{align}
{Z_i^k(\lambda, x)} = 
\hat{X}^i(x)  T - \lambda(k \cdot t^c_i+T),
\end{align}
where $\lambda$ is a positive finite valued variable. From [Theorem 1, Chapter 3, \cite{uclastopping}] we know that an optimal stopping rule exists if the following two conditions are satisfied\footnote{The interpretation of these two conditions is that even with knowing the future the maximum expected reward, or the reward approaching the supremum, is finite.}:
\begin{align}
\text{(C1)} ~~ E\{\sup_k {Z_i^k(\lambda, X^i)}\} < \infty, ~~ \text{(C2)} ~~\lim_{k \rightarrow \infty} {Z_i^k(\lambda, X^i)} \leq {Z_i^{\infty}(\lambda, X^i)}, a.s.
\end{align}
Since we have a finite number of channels and the channel state realization is finite, $\hat{X}^i(x)$ is  finite.  Therefore ${Z_i^{\infty}(\lambda, X^i)} = -\infty$. 
Since $\hat{X}^i(x) $, $\frac{T}{T+t^s_{i+1}}E\{V_{i+1}(X^{i+1})\}$ and $T$ are all finite, (C2) is easily satisfied.  Next define
\begin{eqnarray}
{Z_i(\lambda, x)} = 
\hat{X}^i(x)  T - \lambda T,
\end{eqnarray}
which is again finite.
Therefore we have
$E\{{Z_i(\lambda, X^i)}\} < \infty$,  and $E\{({Z_i(\lambda, X^i)})^2\} < \infty$.
Also noting that ${Z_i(\lambda, X^i)}$ is IID since $X^i$ is IID, by the dominated convergence theorem we have $E\{\sup_k {Z_i^k(\lambda, X^i)}\} < \infty$, verifying (C1). 
The existence is thus established. 

Next we prove that the optimal stopping rule is given by a threshold. Using the principle of optimality [Chapter 2, \cite{uclastopping}] and the results from [Section 4.1,  \cite{uclastopping}] (we refer the reader to [Example 6.2,  \cite{uclastopping}] for further detail), our problem as expressed in (\ref{eqn:transform}) is equivalent to a rate-of-return problem with a reward of stopping given by 
{$Z_i(\lambda, x)$} and a cost of continuation given by 
$\lambda t^c_i$.  The optimal stopping rule at step $k$ 
is given by
\begin{align}
\Delta^s_i =  \{x: {Z_i(\lambda^*, x)} \geq 0\} = \{x: \hat{X}^i(x) \geq \lambda^{*}\}, 
\end{align}
where $\lambda^{*}$ is such that the function $V_k^{*}(\lambda)$, defined recursively as ([Chapter 6,  \cite{uclastopping}])
{$V^*_k(\lambda) = E\{\max \{Z_i(\lambda, x)-\lambda t^c_i,V^*_{k}(\lambda)-\lambda t^c_i\} \}$},  
is evaluated to be zero, i.e., $V_k^{*}(\lambda^{*}) = 0$. To obtain $\lambda^{*}$, we take $V_k^{*}(\lambda^{*}) = 0$ into the above definition and get
$E\left\{\max\{{Z_i(\lambda, x)}, 0\}\right\} = \lambda^{*} t^c_i$,
or equivalently, $\lambda^{*}$ is such that it satisfies
\begin{align}  \label{eqn:threshold_eqn}
E[\hat{X}^i(x)T - \lambda^{*}T]^{+}= \lambda^{*}t^c_i,
\end{align}
%
%
which is the same as (\ref{sol}).  This completes the proof of the form of the threshold.
%
It remains to show that a unique solution exists to (\ref{sol}). Denote by
$\mathcal D(\lambda) = E[\hat{X}^i(x) - \lambda]^{+}- \frac{\lambda t^c_i}{T}$.
It is not hard to verify that $\mathcal D(\lambda)$ is a continuous and strictly decreasing function of $\lambda$.  Furthermore, we have $\mathcal D(\lambda = 0) = E[\hat{X}^i(x)]^{+} > 0$ since all channel states are positive, and $\mathcal D(\lambda) \rightarrow -\infty$ as $\lambda \rightarrow \infty$.
Therefore there is a unique solution to $\mathcal D(\lambda) = 0$, i.e., the threshold exists and is unique, completing the proof.
\end{proof}

In practice, to calculate this threshold, we define
\begin{eqnarray}
c_i = \frac{T}{T+t^s_{i+1}}E\{V_{i+1}(X^{i+1})\}.
\end{eqnarray}
Re-writing (\ref{sol}) in the original random variables, we have 
\begin{eqnarray}
&& E\left[\max \left\{X^i,\frac{T}{T+t^s_{i+1}}E\{V_{i+1}(X^{i+1})\} \right\} T - \lambda T\right]^+ \nonumber \\
&=& E\left\{\max \{X^i T-\lambda T,0\}|X^i > c_i\right\}\cdot P(X^i > c_i) +E\left\{\max \{c_iT-\lambda T,0\}|X^i \leq c_i \right\} \cdot P(X^i \leq c_i) \nonumber \\
&= & \lambda t^c_i~.
\end{eqnarray}
If the solution $\lambda^{*}<c_i$, then it has to satisfy
$\int_{c_i}^{\bar{X}^i} (x-\lambda)f_{X^i}(x)dx + (c_i-\lambda)\cdot P(X^i \leq c_i) = \lambda {t}^c_i /T$, and thus can be obtained by
$\lambda^{*} = \frac{\int_{c_i}^{\bar{X}^i}xf_{X^i}(x)dx+c_i\cdot P(X^i \leq c_i)}{1+{t}^c_i/T}$ and verifying that the resulting $\lambda^{*} < c_i$.
If the solution $\lambda^{*}\geq c_i$, then
%
it must satisfy
$\int_{\lambda}^{\bar{X}^i} (x-\lambda)f_{X^i}(x)dx = \lambda {t}^c_i/T$, and the solution may be obtained using  
$\lambda^{*} = \frac{\int_{\lambda^{*}}^{\bar{X}^i} x f_{X^i}(x)dx}{P(X^i \geq \lambda^{*}) + {t}^c_i/T}$ and verifying\footnote{This function is a fixed point equation which could be solved by iterative methods as in \cite{Zheng:2009:DOS:1669334.1669354}.} that the resulting $\lambda^{*} \geq c_i$. 



\begin{remark}
The quantity $c_i$ defined above is the expected reward of SWITCH, while $\lambda^*$ is the threshold for making a decision between the set \{STOP, SWITCH\} and STAY.  The optimal policy given in the above  theorem is illustrated in Figure \ref{ill_policy}, which can be viewed as a sequence of two YES/NO questions used in decision making involving two thresholds.  (1) If $\lambda^*< c_i$, then the optimal decision is either STOP or SWITCH depending on whether $x> c_i$.  If the current condition is very good ($x>c_i$) then the decision is STOP; otherwise SWITCH.  In this case the reward from switching is sufficiently good that we will never consider STAY.  (2) If $\lambda^* > c_i$, then the optimal decision is either STOP or STAY depending on whether $x>\lambda^*$. In this case the reward from switching is inferior so that SWITCH is not an option. This policy will be referred to as a {\em nested stopping} policy.
\begin{figure}[h!]
    \centering
        \includegraphics[width=0.4\textwidth,height=0.2\textwidth]{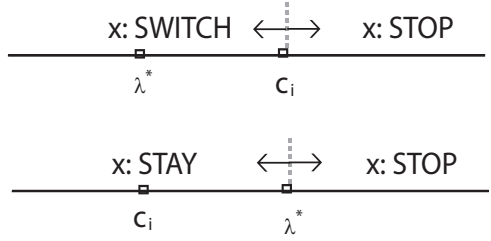}
    \caption{Illustration of the decision process}\label{ill_policy}
\end{figure}

\end{remark}

\subsection{Properties of the nested stopping policy}

We next investigate a number of properties of the multiuser multichannel system as a result of the above nested stopping policy.  We start by examining its effect on the traffic load vector $\mathbf{G}$. 
Unless otherwise noted, all proofs can be found in the appendix. 

\begin{lemma}[Monotonicity of the value function]
Consider two traffic load vectors $\mathbf{G}$ and $\mathbf{G^{'}}$ where $G_{i} \leq G^{'}_{i},\forall i \in \Omega$.  Denote the corresponding sets of value functions by $V$ and $V^{'}$, respectively.  Then we have
$E\{V_{i}\} \geq E\{V^{'}_{i}\}, \forall i \in \Omega$.  
\label{lemma2}
\end{lemma}
This lemma conveys the intuition that when the load increases, competition increases leading to longer delays.
Thus the expected throughput decreases in general.  
We next establish the stability of the system under the nested stopping policy, starting with an assumption. 

\begin{assumption} 
No channel is dominant, i.e., there is no single channel that will attract all arrivals under the nested stopping policy. \label{assupsta}
\end{assumption}
This assumption excludes the extreme case where a single channel is of far better quality (e.g., very high data rate) that even considering the cost in competition it is beneficial to always switch to this channel, no matter the  conditions observed in the other channels. 

\begin{lemma}[Ergodicity of the arrival process]
The arrival processes are ergodic under the nested stopping policy and Assumption \ref{assupsta}. \label{staiid}
\end{lemma}



\begin{lemma}[Load balance]
{We have $\frac{\partial G_i}{\partial G} \geq 0,\forall i \in \Omega$ under the nested stopping policy}.\label{iidlg}
\end{lemma}
In other words, if the total traffic load increases, the input/load to each channel is non-decreasing.  This property combined with the monotonicity (Lemma \ref{lemma2}) leads to the following stronger monotonicity result on the value function; the proof is trivial and thus omitted.
\begin{lemma}[Strong monotonicity] $E\{V_{i}\},i \in \Omega$, are all non-increasing functions of $G$.
\end{lemma}

We next analyze the impact of transmission time $T$ and address the question whether by reserving more time for a single transmission users gain in average throughput.

\begin{lemma}[Impact of $T$]
$E\{V_i\},i \in \Omega$, are all non-decreasing functions of $T$. \label{iidt}
\end{lemma}
This result reflects the intuition that once a user finds a good transmission condition, it is beneficial for it to be able to use it for a longer period of time. However, practically $T$ cannot be made too large due to the channel coherence time: the channel condition will likely change over a large period $T$. 
\section{Optimal access policy under the Markovian channel model}\label{sec:markov}

This section presents a parallel effort to the previous section, under the assumption that the channel conditions evolve over time as a Markov chain.

\subsection{Uniqueness of the optimal strategy}

Denote the state space of channel $i$ by $S_i$, and the single-step (over one unit of time) state transition probability by $\mathcal P_i(y|x)$, $x, y\in S_i$.  The $k$-step transition probability is denoted by $\mathcal P^k_i(y|x)$.  The value function representing the maximum average throughput given the current condition at stage $i$ is given by the following. 
\begin{align}
V_{i} (x) = \max \left\{\hat X^i(x), \frac{T}{t^c_i+T}\cdot\sum_{y \in S_i} \mathcal P^{t^c_i}_i(y|x)\cdot V_{i}(y) \right\}\label{m1}
\end{align}
where $\hat X^i(x)$ follows the same definition as in the IID case. We make the following approximation. When $1/T$ is sufficiently small\footnote{This is possible since $T$ is an integer multiple of an arbitrary time unit, which can be made very small.  The only restriction is that we have taken a single time unit to be the time it takes to transmit a  control packet, so this assumption simply implies that a data transmission is much longer than a control transmission, which is typically true.}, we have
$\frac{T}{t^c_i+T} = \frac{1}{1+t^c_i/T} \approx (\frac{1}{1+1/T})^{t^c_i}$.
Denote $\beta = \frac{1}{1+1/T}$ and we arrive at the following an approximated value function
\begin{align}
V_{i} (x) = \max \left\{\hat X^i(x), \beta^{t^c_i}\cdot \sum_{y \in S_i} \mathcal P^{t^c_i}_i(y|x)\cdot V_{i}(y) \right\}. \label{m1}
\end{align}
Denote by $\mathcal U= \{\text{\textbf S, \textbf C}\}$ the set of two actions, stopping and continuation, where the stopping action ${\textbf S}$ bundles STOP and SWITCH into a single action, i.e., $\textbf{S} = \{\mbox{STOP, SWITCH}\}$ due to the definition of $\hat{X}^i(x)$ and as in the IID case, and the continuation action $\textbf{C} = \{\mbox{STAY}\}$. 
Then the above can be re-written as 
\begin{align}
V_i(x) = \max_{u \in \mathcal U} \left\{r(u,x)+\beta^{t^c_i} \cdot \sum_{y \in S_i}\mathcal P_i^{u,t^c_i}(y|x)\cdot V_i(y)\right\} \label{m2}
\end{align}
where
$r(\text{\textbf S},x) = \hat X^i$, $r(\text{\textbf C},x) = 0$, $\mathcal P^{\text{\textbf S},t^c_i}_i(y|x) = 0$, and $\mathcal P^{\text{\textbf C},t^c_i}_i(y|x) = \mathcal P^{t^c_i}_i(y|x)$.
\begin{theorem}
The set of Equations (\ref{m1}) or equivalently (\ref{m2}) have a unique solution.
\end{theorem}
%
Our proof is based on the contraction mapping theorem \cite{Kumar:1986:SSE:40665} and the next lemma. 
\begin{lemma}
\label{52}
Let $\mathcal F$ be the class of all functions $v:\{1,2,...,S\} \rightarrow \mathcal R$. Define norm $||v|| := \sum_{x \in S}|v(x)|$ and a mapping $\mathcal T:\mathcal F \rightarrow \mathcal F$ by
$$
(\mathcal Tv)(x) := \max_{u \in \mathcal U} \{r(u,x)+\eta \cdot \sum_{y \in S}v(y)\cdot \mathcal P^u(y|x)\}
$$
$0 < \eta < 1$; then $\mathcal T$ is a contraction. \label{lemma_mc}
\end{lemma}
The next result also immediately follows; the proof is omitted for brevity. 
\begin{corollary}[Threshold policy]
The optimal stopping rule reduces to a threshold policy.
\end{corollary}
\begin{remark}
As may be expected, this threshold policy works in a way very similar to the IID case (only the numerical calculation differs): at stage/channel $i$, there is a SWITCH reward $c_i$ (expected throughput by switching away from $i$) and $\lambda^*$ by staying on the same channel.  The optimal decision is then based on the relation between $\lambda^*$ and $c_i$.
\end{remark}

\subsection{Properties of nested stopping policy}
We can similarly obtain a number of properties of the multiuser multichannel system as a result of the nested stopping policy under the Markovian model.  

\begin{theorem} [Monotonicity]
\label{53}
$E\{V_i\},i \in \Omega$ are all non-increasing functions of G.
\end{theorem}

Following the above result we can derive similar properties of the nested stopping policy in the Markovian case as in the IID case, including ergodicity of the arrival processes, load balance and the non-increasing value functions in $T$. The proof of these are omitted due to brevity and their similarity to those in the IID case. 

\section{Numerical Results}\label{simu}

\subsection{The IID channel model}
We first consider a scenario of five independent channels with their channel condition (taken to be the instantaneous transmission rate measured in bytes per time unit) exponentially distributed over a finite range, with average rates given by
$
\{1/0.4,1/0.6,1/0.5,1/0.3,1/0.2\}
$. A single transmission period is set to $T=40$ time units. 
The level of contention/congestion measured by $t^c_i$ and $t^s_i$ (measured in time units) as a function of  load $G$ (measured in packet per unit time) is illustrated in Table \ref{iidtable} for channel 1.  These quantities are rounded off to the nearest integers when used in computing the optimal policy. \rev{We set packet length to be 1024 Bytes.}  
\begin{table}[!h]
\small
\begin{center}
  \begin{tabular}{ | c | c | c | c | c | c |}
    \hline
    Load & 0.1 & 0.2 & 0.3 & 0.4 & 0.5\\ \hline
    $t_i^c$ & 10.8 & 13.2 & 14.4 & 15.2 &15.8\\ \hline
    $t_i^s$ & 13.1 & 15.8 & 17.3 & 18.4 & 19.4\\ \hline
  \end{tabular}
\end{center}
    \caption{Contention levels}\label{iidtable}
\end{table}
\begin{figure}[!h]

\subfigure[Non-opportunistic v.s. nested stopping policy (channel 1) ]{
                \centering
                \includegraphics[width=0.5\textwidth,height=0.25\textwidth]{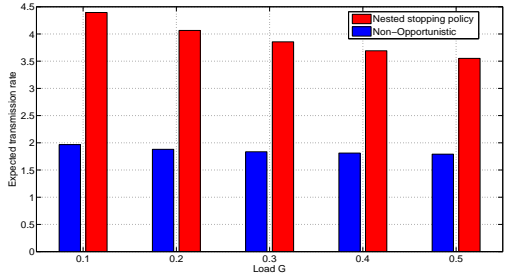}\label{fig1}
        }
\hspace{-1.3cm}
        \subfigure[Transmission rate w.r.t. $G$]{
                \centering
                \includegraphics[width=0.5\textwidth,height=0.25\textwidth]{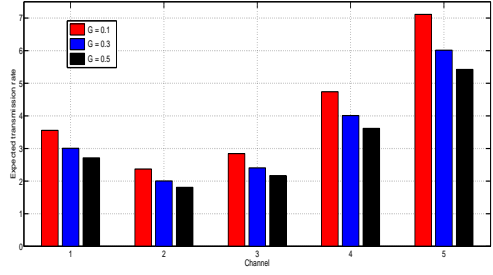}\label{fig2}
        }\\
        \subfigure[Transmission rate w.r.t. $T$ ($G=0.5$)]{
                \centering
                \includegraphics[width=0.5\textwidth,height=0.25\textwidth]{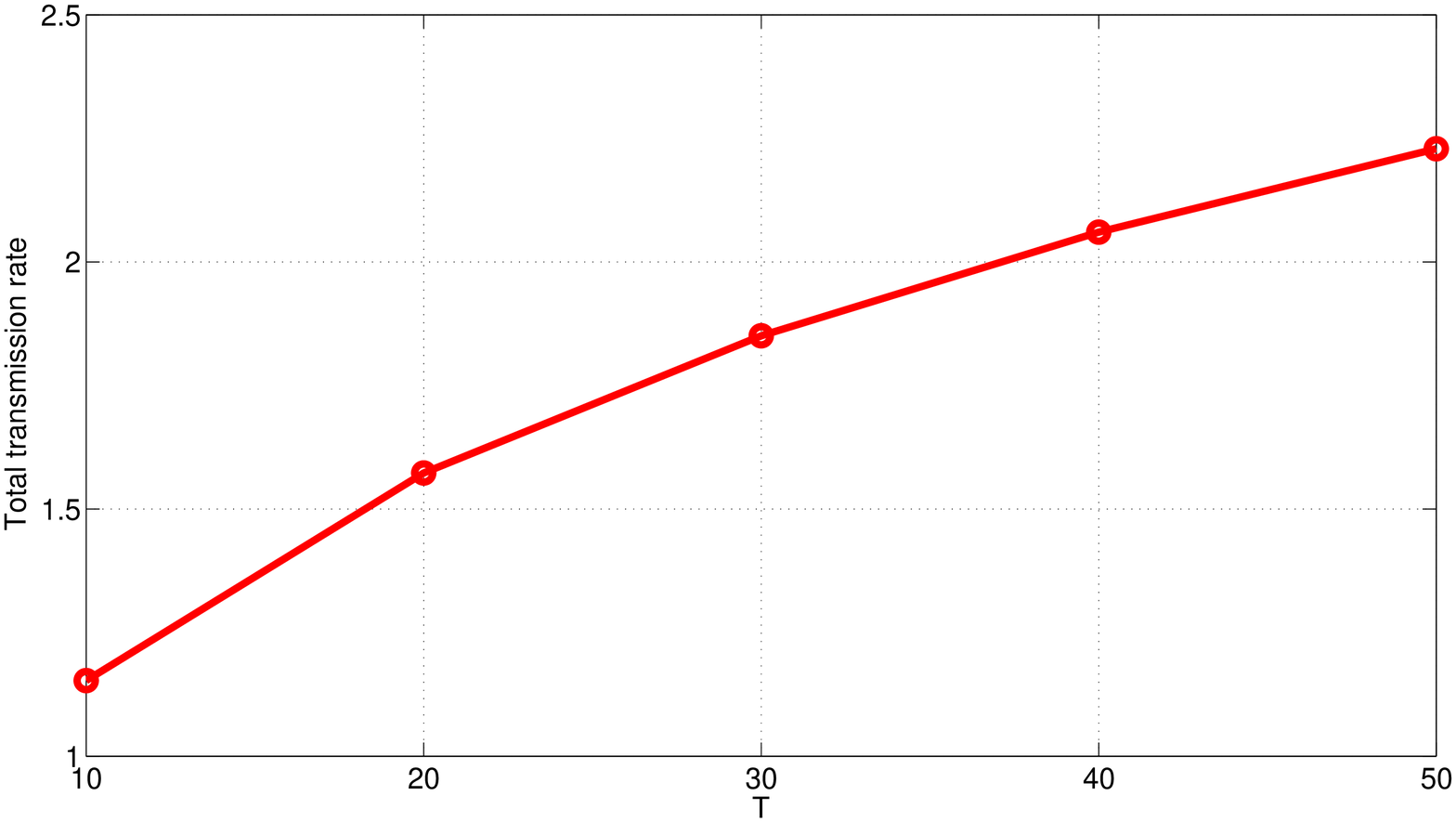}\label{fig3}
        }
\hspace{-1.3cm}
                \subfigure[Transmission rate w.r.t. number of channels ($G=0.5$)]{
                \centering
                \includegraphics[width=0.5\textwidth,height=0.25\textwidth]{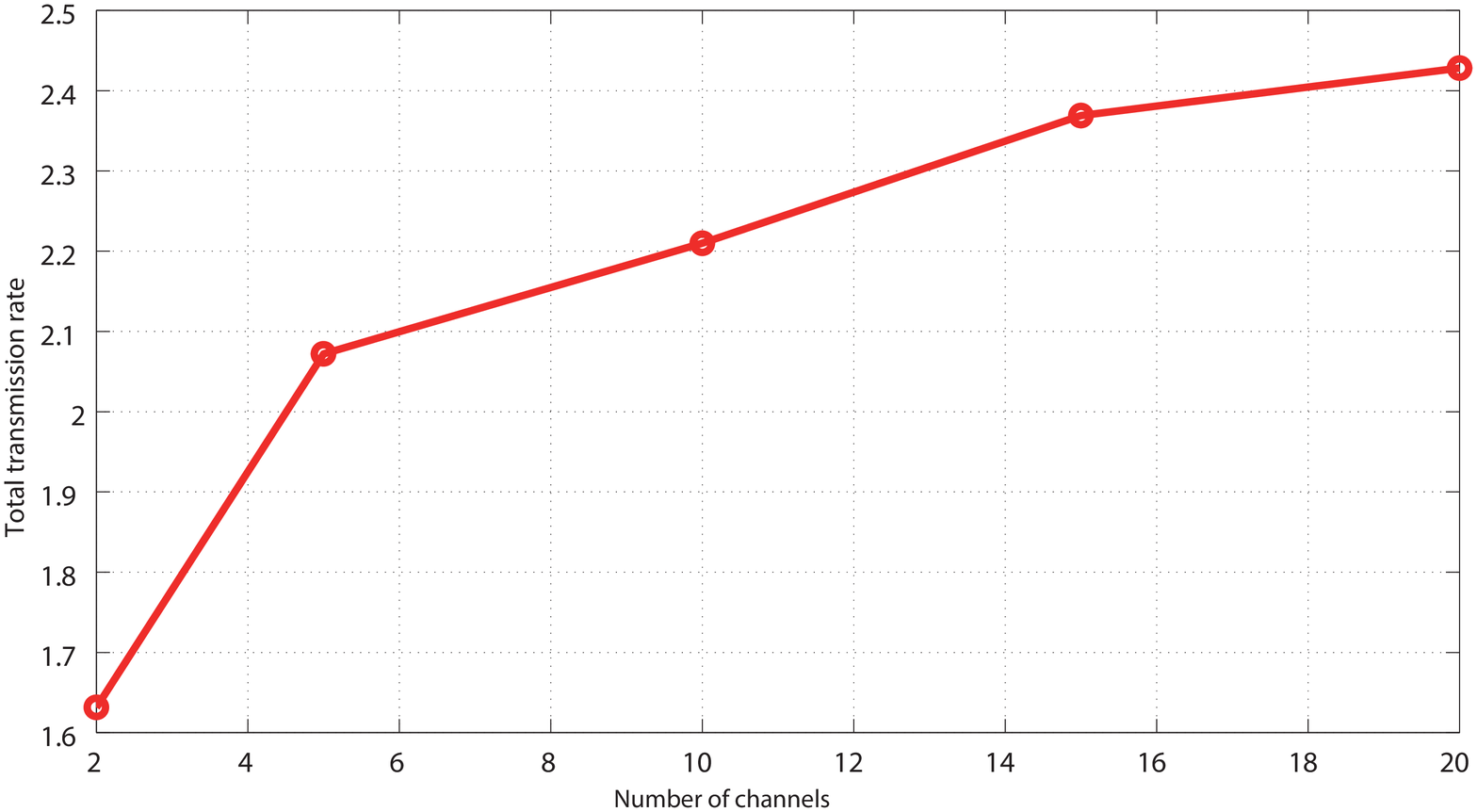}\label{fig4}
        }
\caption{Performance of the nested stopping policy under the IID channel model}
\end{figure}

In Figure \ref{fig1} we compare the nested stopping policy with a standard random access policy in which a user randomly selects a channel to use, followed by competing for channel access using IEEE 802.11 type of random access scheme (channel 1 is shown to illustrate).  Fig. \ref{fig1} shows that our nested stopping policy clearly helps improve the system performance w.r.t. average data rate \rev{(measured in average bytes per time unit transmitted, or average throughput)}; note however that this improvement decreases with the increase in load $G$, suggesting that the increase in congestion dampens the positive effect of opportunistic channel access. Fig. \ref{fig2} shows this more clearly the average data rate under different loads for each channel. Fig. \ref{fig3} shows that the throughput increases in the data transmission time $T$ as we have characterized, and Fig. \ref{fig4} shows that it also increases in the number of channels
(the simulation is done by adding channels with same statistics as given for the initial five), as the contention in each channel reduces.

Next we show the decision table for the optimal actions conditioned on continuation (STAY or SWITCH) for each channel (in this specific experiment we consider a user starts from channel 1).
\begin{table}[!h]
\small
\begin{center}
\begin{tabular}{ | c | c | c | c | c | c |}
    \hline
    Load & Ch 1 & Ch 2 & Ch 3 & Ch 4 & Ch 5\\ \hline
    0.05 & STAY & SWITCH & SWITCH & SWITCH &STAY\\ \hline
    0.1 & STAY & SWITCH & SWITCH & STAY &STAY\\ \hline
    0.3 & STAY & SWITCH & SWITCH & STAY &STAY\\ \hline
    0.5 & STAY & SWITCH & SWITCH & STAY &STAY\\ \hline
  \end{tabular}
\end{center}
  \caption{Decision of IID channels with different arrival rate}
\end{table}
As can be seen, channels 2 and 3 are of low quality so that the general decision is to switch away rather than waiting on the same channel if the decision is not to transmit immediately.  For channel 4, we see that the tendency to stay increases when the load is high due to the higher cost in switching than staying.  The decision to stay in channel 1 is more interesting: even though better average throughput may be obtained in channels 4 and 5, the cost in doing so is considerable as it has to go through channels 2 and 3.  By contrast, there is a SWITCH decision in channel 4 even though channel 4 is on average a better channel than channel 1.  



We also consider a more practical AWGN wireless channel model considering both propagation loss and shadowing effects. The transmission rates  are given by the Shannon capacity formula for AWGN channels: 
$R = \log(1+\rho | h |^{2}) \textrm{ nats/s/Hz}$,
where $h$ denotes the random channel gain with a complex Gaussian distribution. Moreover, the cdf of transmission rate is given by
$F_R(r) = 1 - \text{exp}(-\frac{\text{exp}(r)-1}{\rho}), r \geq 0$.
Consider a scenario with five channels with average SNR $\rho$ given by $10, 25, 20, 30, 10$, respectively, and similar performance results are observed as shown in Figures \ref{wlfig1} and \ref{wlfig2}.

\begin{figure}[!h]
\centering
\subfigure[Non-opportunistic v.s. nested stopping policy (channel 1)]{
                \centering
                \includegraphics[width=0.5\textwidth,height=0.25\textwidth]{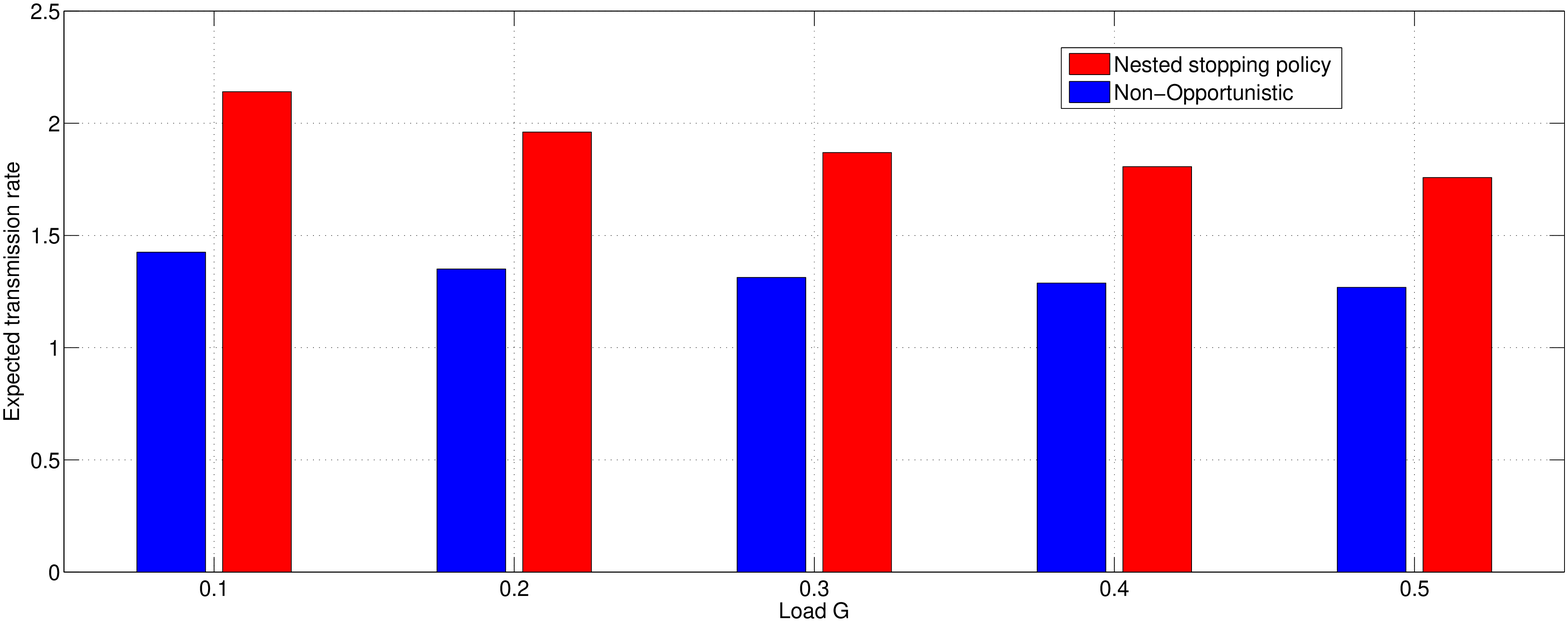}\label{wlfig1}
        }
\hspace{-1.3cm}
        \subfigure[Transmission rate w.r.t. $G$]{
                \centering
                \includegraphics[width=0.5\textwidth,height=0.25\textwidth]{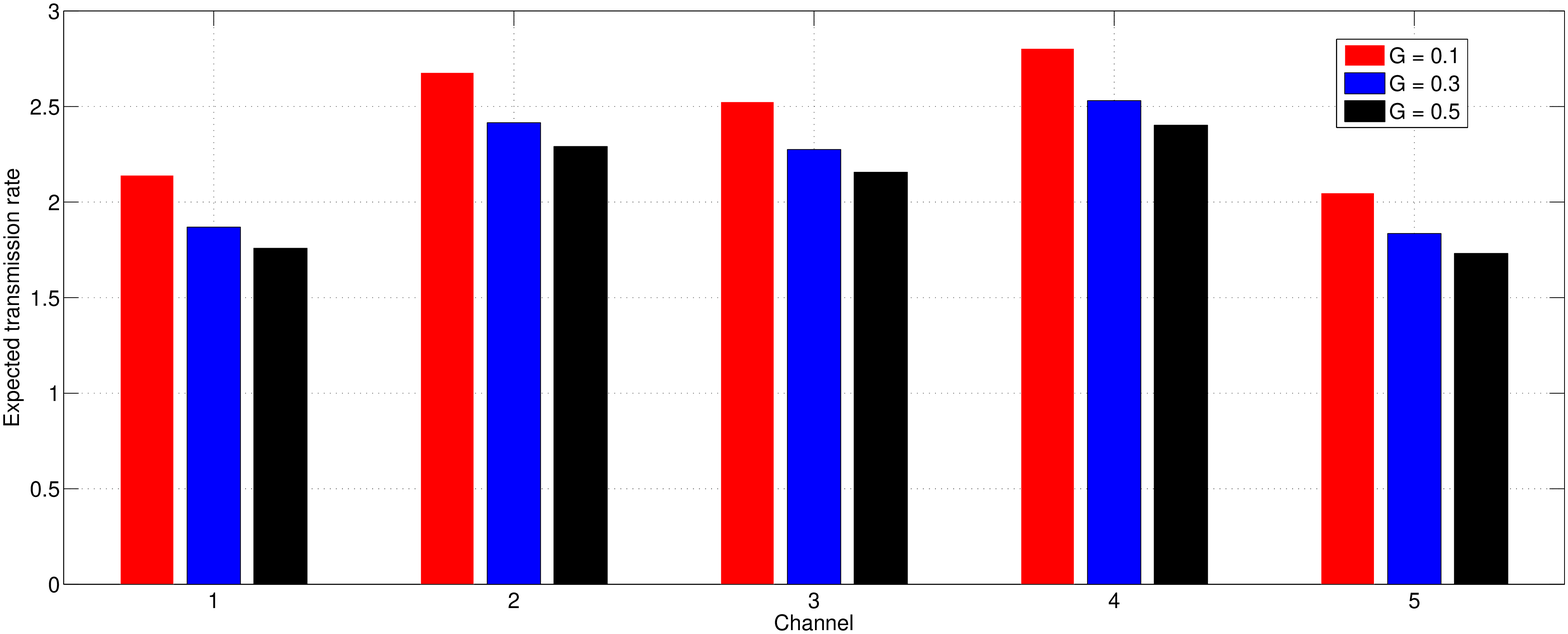}\label{wlfig2}
        }
\caption{Performance comparison under an AWGN wireless link model}
\end{figure}

\subsection{The Markovian channel model}

We now simulate the nested stopping policy under a Markovian channel model. We model all five channels' state (again taken to be the instantaneous transmission rate in bytes per time unit) change as a birth-death chain with five states and the associated transition probabilities given as follows:
\begin{align}
&\mathcal P_{k}(i+1|i) = 0.8, \mathcal P_{k}(i-1|i) = 0.2, 1 < i < 5, 1 \leq k \leq 5\\
&\mathcal P_{k}(2|1)= 0.8, \mathcal P_{k}(1|1) = 0.2,\mathcal P_{k}(5|5) = 0.8, \mathcal P_{k}(4|5) = 0.2, 1 \leq k \leq 5
\end{align}
For each channel the rewards increase in state indices, and are given in Table \ref{chrwd}. Transmission time is again set to be $T=40$ time units. 

\begin{table}[!h]
\small
\begin{center}
  \begin{tabular}{ | c | c | c | c | c | c |}
    \hline
    States & Ch 1 & Ch 2 & Ch 3 & Ch 4 & Chl 5\\ \hline
    1 & 10 & 15 & 5 & 10 & 5\\ \hline
    2 & 20 & 20 & 10 & 20 &10\\ \hline
    3 & 30 & 45 & 15 & 30 &15\\ \hline
    4 & 40 & 60 & 20 & 40 &20\\ \hline
    5 & 50 & 75 & 25 & 50 &25\\ \hline
  \end{tabular}
\end{center}
    \caption{Reward table for Markovian channels with different states}\label{chrwd}
\end{table}

\begin{figure}[!h]
\centering
\subfigure[Non-opportunistic v.s. nested stopping policy (channel 1)]{
                \centering
                \includegraphics[width=0.5\textwidth,height=0.25\textwidth]{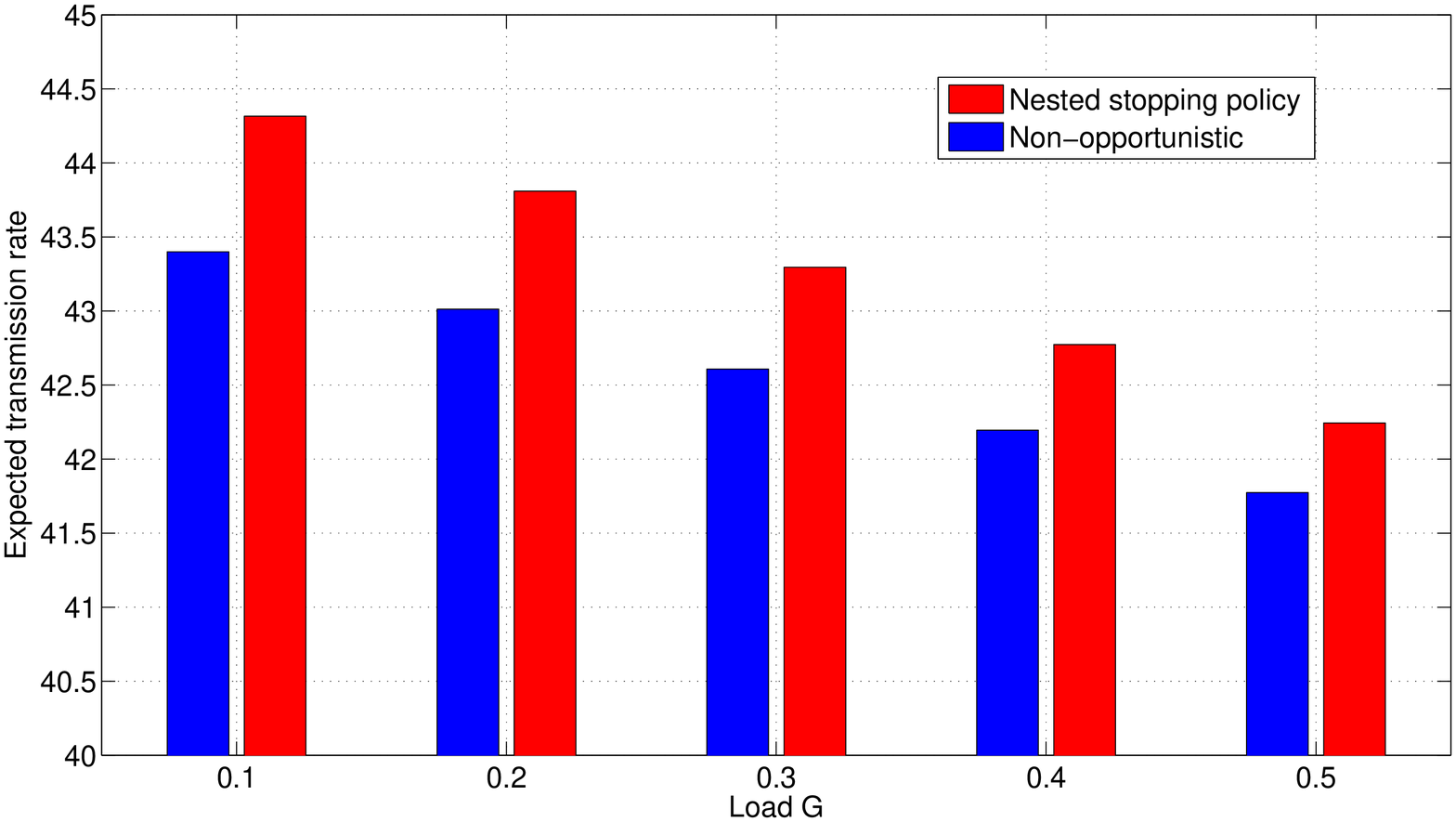}\label{mcfig4}
        }
\hspace{-1.3cm}
        \subfigure[Transmission rate w.r.t. $G$ ]{
                \centering
                \includegraphics[width=0.5\textwidth,height=0.25\textwidth]{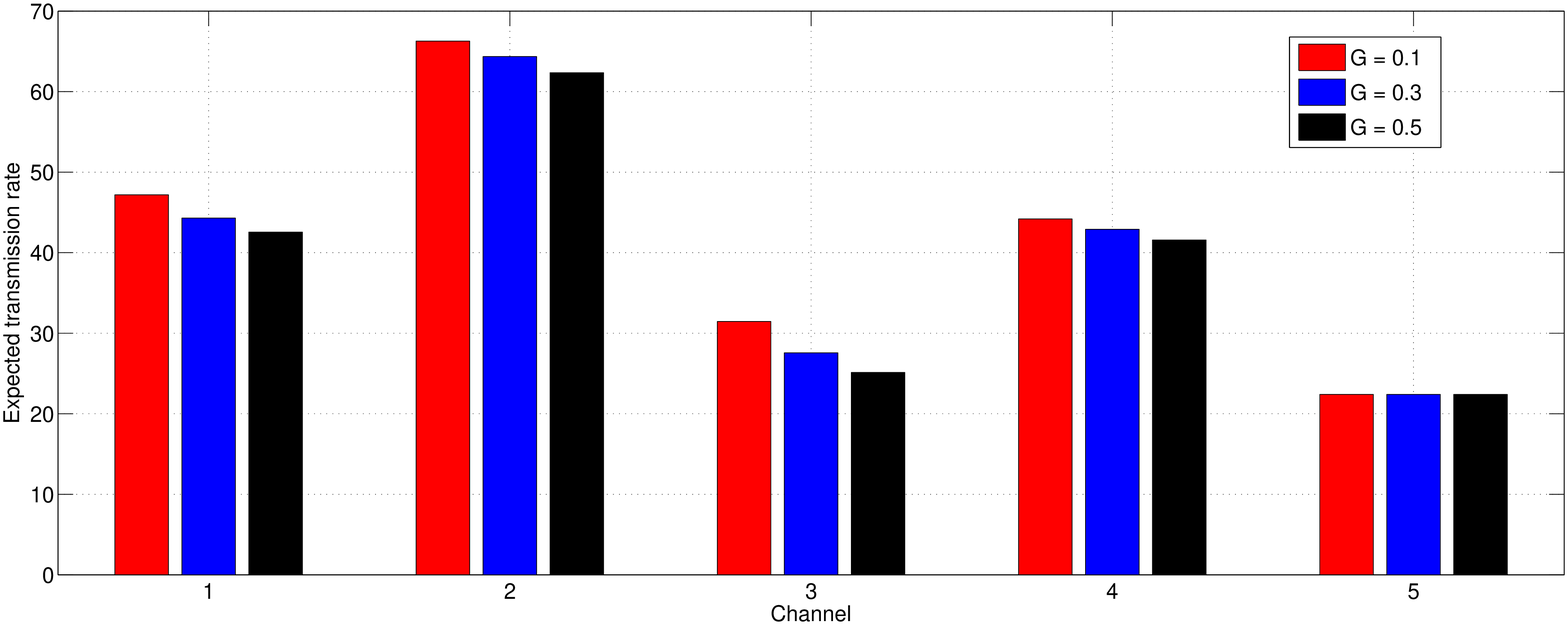}\label{mcfig6}
        }
        \subfigure[Transmission rate w.r.t. $T$ ($G = 0.1$)]{
                \centering
                \includegraphics[width=0.5\textwidth,height=0.25\textwidth]{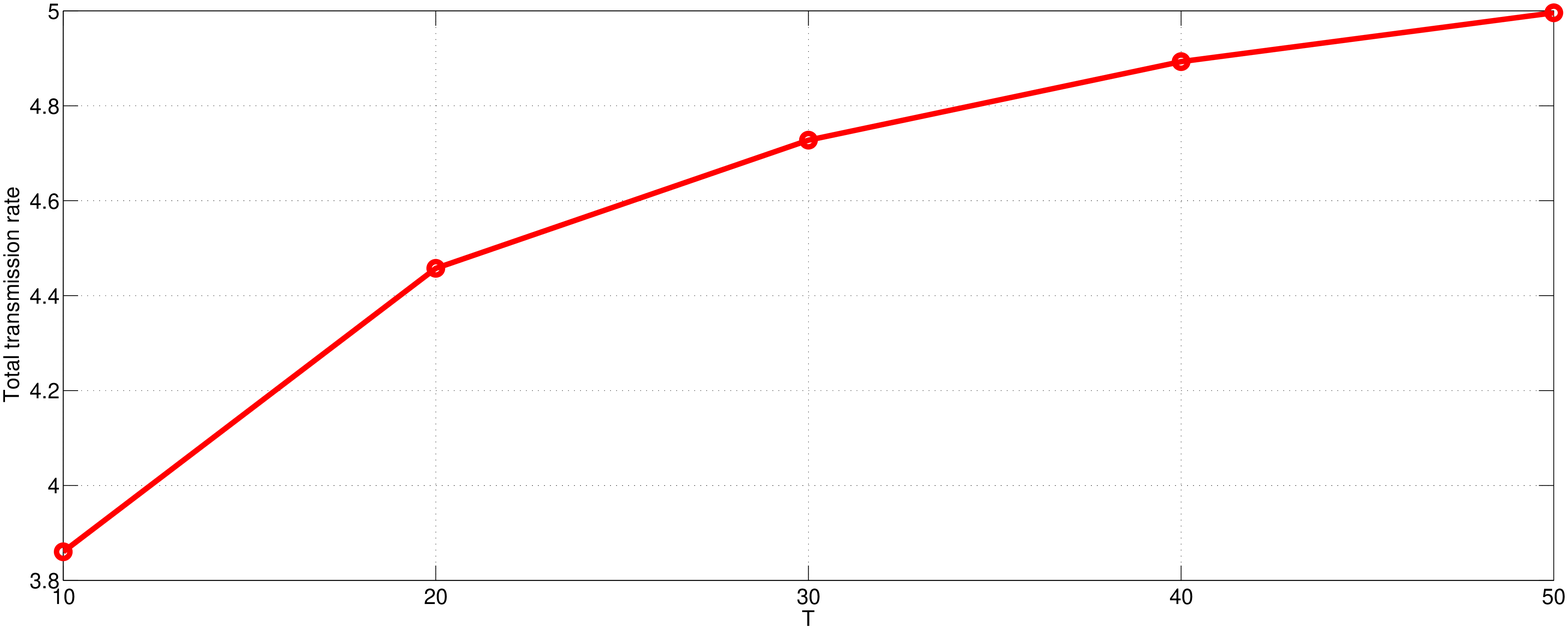}\label{mcfig7}
        }
\hspace{-1.3cm}
        \subfigure[Transmission rate w.r.t. number of channels ($G = 0.1$)]{
                \centering
                \includegraphics[width=0.5\textwidth,height=0.25\textwidth]{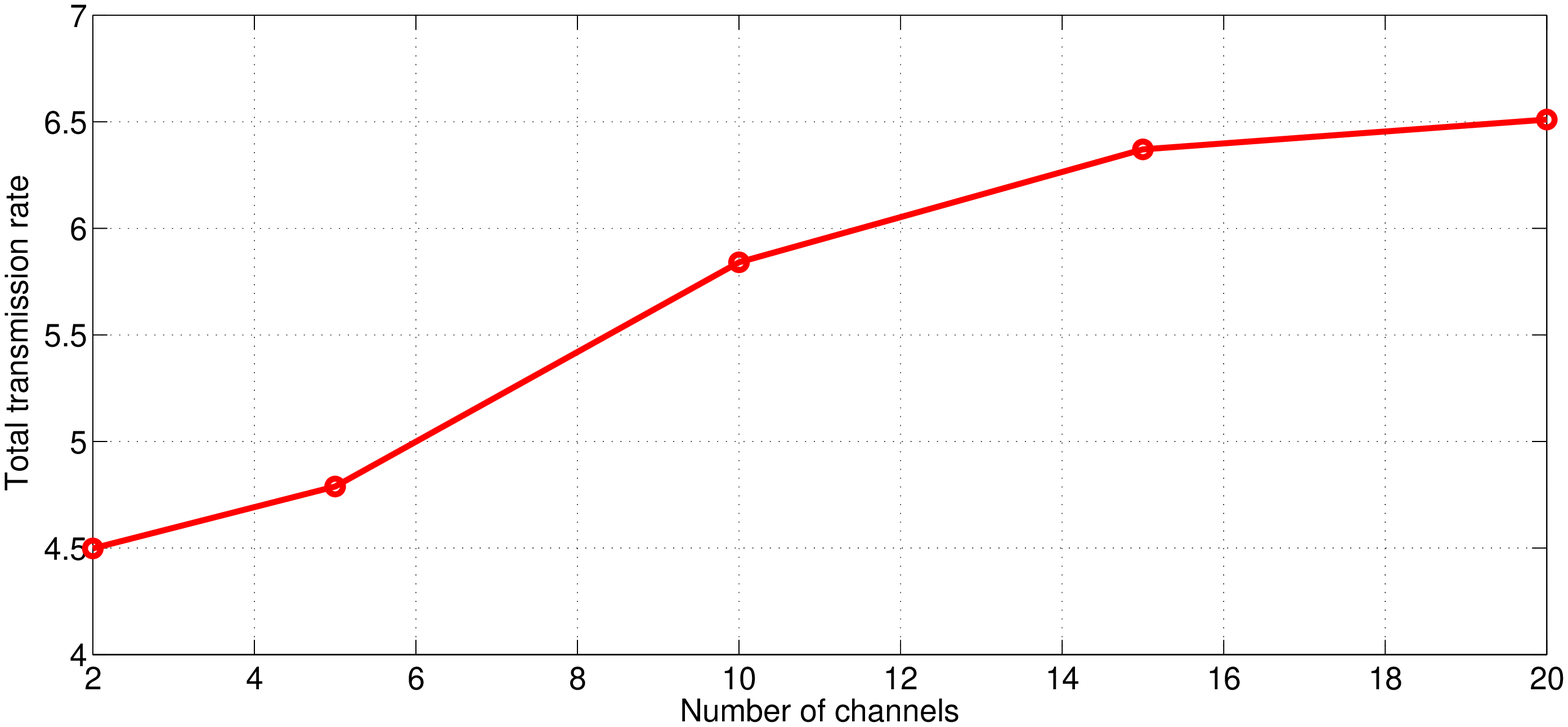}\label{mcfig7}
        }
\caption{Performance of the nested stopping policy under the Markovian channel model}
\end{figure}

The performance results are shown in Figs. \ref{mcfig4} and \ref{mcfig6}.  With a slow changing model, users are more likely to transmit using the currently sampled rate instead of releasing it and waiting for a future opportunity while risking another contention period.  In other words, opportunistic access in this case provides only marginal improvement over the non-opportunistic method.  
We also provide the decision table in this case in Table \ref{table2}.
\begin{table}[!h]
\small
\begin{center}
  \begin{tabular}{ | c | c | c | c | c | c |}
    \hline
    States & Ch1 & Ch 2 & Ch 3 & Ch 4 & Ch 5\\ \hline
    1 & SWITCH & STAY & SWITCH & STAY &STAY\\ \hline
    2 & SWITCH & STAY & SWITCH & STOP &STAY\\ \hline
    3 & SWITCH & STOP & STOP & STOP &STOP\\ \hline
    4 & STOP & STOP & STOP & STOP &STOP\\ \hline
    5 & STOP & STOP & STOP & STOP &STOP\\ \hline
  \end{tabular}
\end{center}
    \caption{Decision table for Markovian channels with different states}\label{table2}
\end{table}
Similar observations are made here: when the channel condition is good enough, the user would choose to transmit immediately (STOP);  the SWITCH decision is associated with poor conditions and when a user hopes to get much better conditions in the next channel; the STAY decision is made on a reasonably good channel and when there is limited prospect of getting better condition in the next channel. 

\subsection{Channel sensing order \& the ``no-recall'' approximation}
We next examine the effect of selecting different sequence of channels to use.
As discussed earlier, with multiple users ($m \geq 2$) it is very challenging to either jointly determine optimal sensing orders for all users involved in a cooperative setting, or determine the equilibrium sensing orders selected by selfish individuals in a non-cooperative setting (e.g., \cite{5738213,5272267}). 
For this reason in our analysis we have assumed that each user follows a fixed (which can be randomly chosen) order.  We now compare this choice where each user randomly picks a sequence in the IID case 
%
with a greedy sensing order in which users sense channels ordered in decreasing mean rewards. This comparison is shown in Fig. \ref{socomp}; it is clear that it is far better for each users to sense in a different order especially when the load is high. 


\begin{figure}[!h]
\centering
\subfigure[Channel sensing order comparison]{
                \centering
                \includegraphics[width=0.5\textwidth,height=0.25\textwidth]{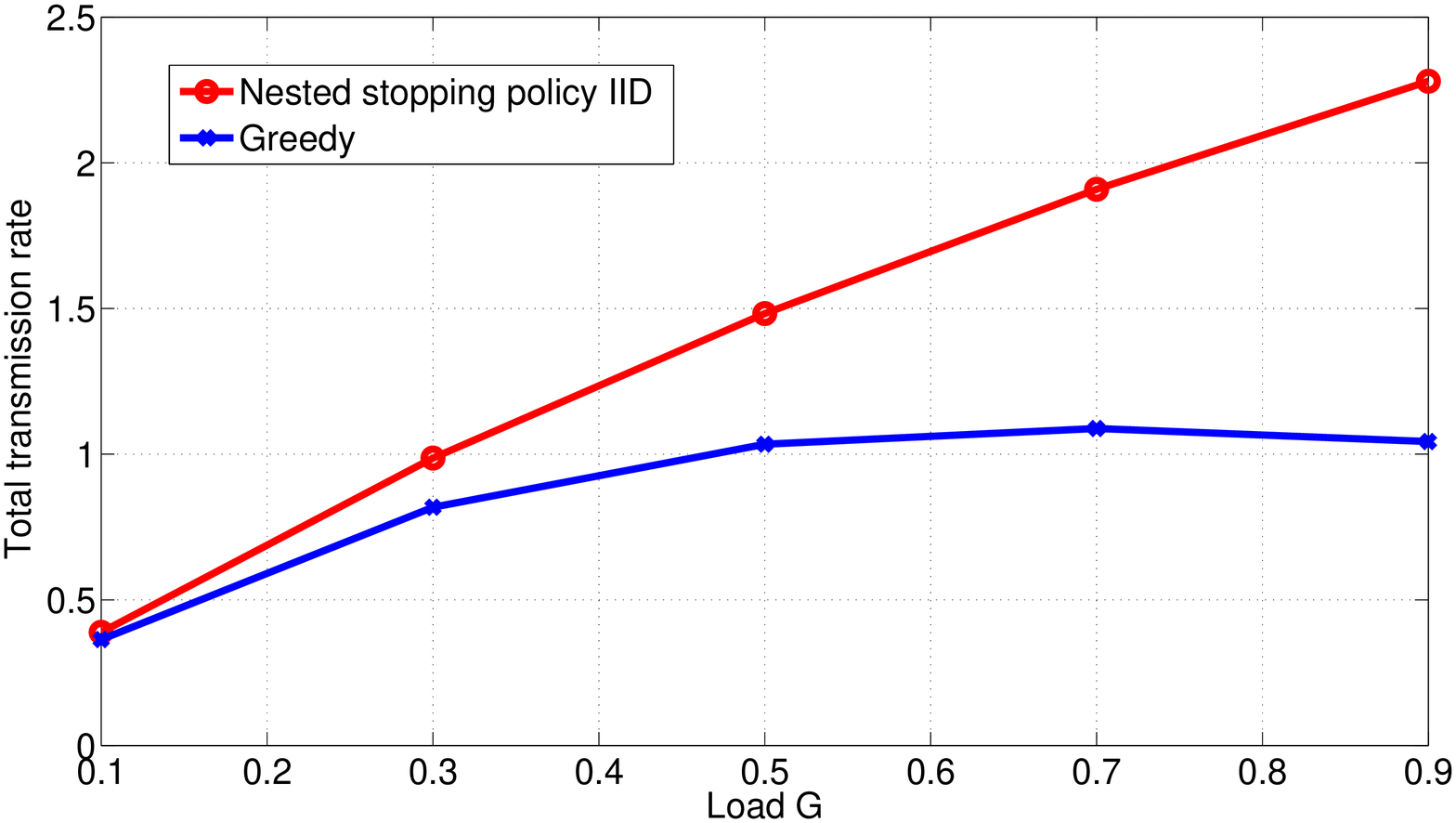}\label{socomp}
        }
\hspace{-1.3cm}
        \subfigure[Performance of approximation model]{
                \centering
                \includegraphics[width=0.5\textwidth,height=0.25\textwidth]{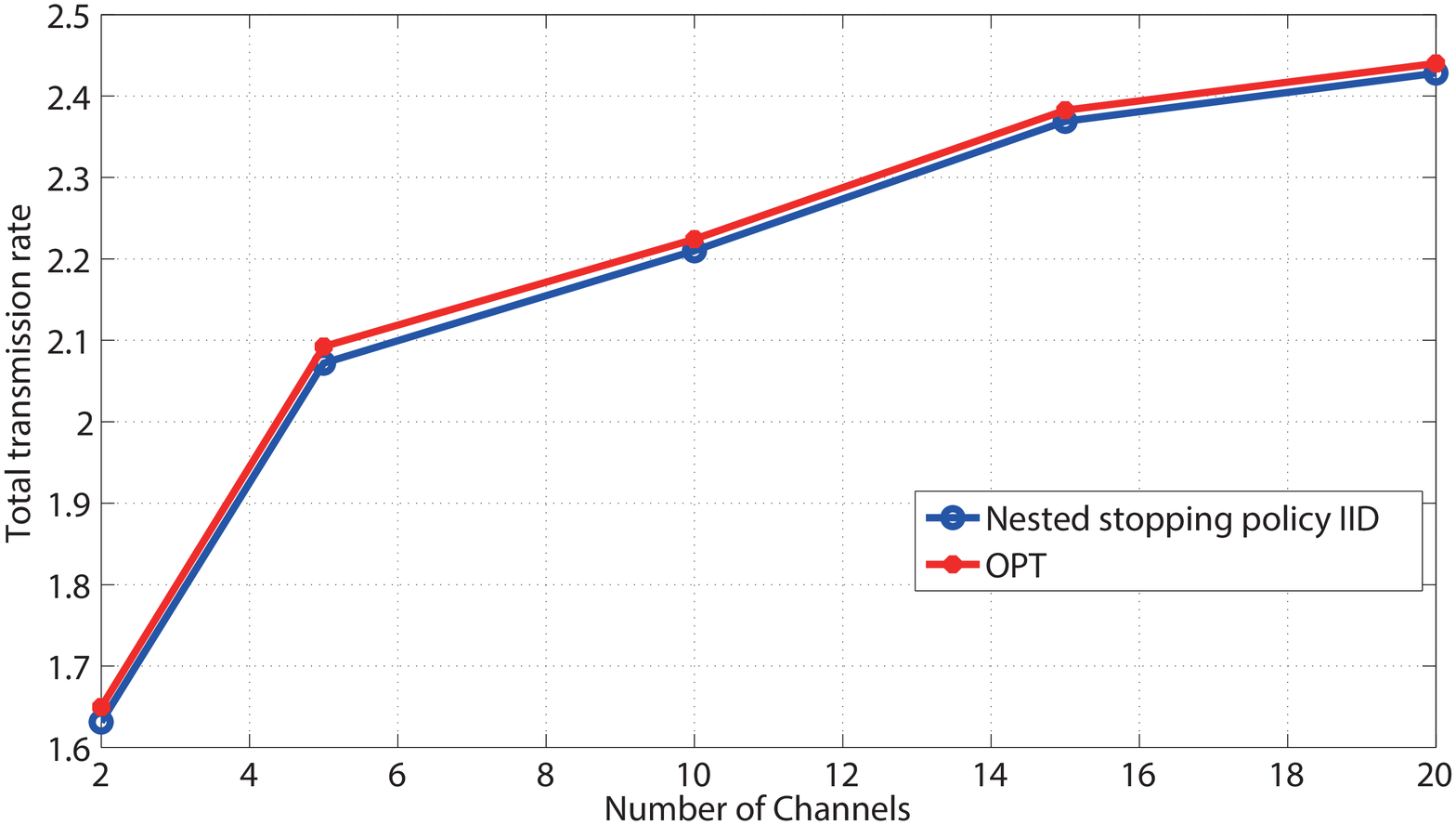}\label{optcomp}
        }

\end{figure}

We end this section by investigating the effect of the ``no-recall'' approximation introduced in Section \ref{model} and adopted in our analysis, by comparing it with the exact optimal solution.  We show this in the IID case in Fig. \ref{optcomp}; we see that this approximation has very little effect on the system performance.


\section{Conclusion}\label{conclusion}
In this paper we considered the collective switching of multiple users over multiple channels.  In addition, we considered finite arrivals.  Under such a scenario, the users' ability to opportunistically exploit temporal diversity (the temporal variation in channel quality over a single channel) and spectral diversity (quality variation across multiple channels at a given time) is greatly affected by the level of congestion in the system. 
We investigated the optimal decision process under both an IID and a Markovian channel models, and evaluate the extent to which congestion affects potential gains from opportunistic dynamic channel switching.
\begin{spacing}{0}
\bibliography{myref}
\end{spacing}
\appendices
\section{Proof of Lemma \ref{lemma2}: monotonicity of value function in $\mathbf{G}$}

\label{appendix-B}
We prove this by induction. When $i = N$, i.e., the last stage, we have
$
\lambda^*_N \bar{t}^c_N = \int_{\lambda^*_N}^{\bar{X}^N} (x-\lambda^*_N)f_{X^N}(x)dx
$, 
where $\bar{t}^c_i = t^c_i/T$.  
As $\bar t^c_N$ is a non-decreasing function in $G_N$, it is also non-decreasing in $\mathbf{G}$. Thus with the increase in $\bar t^c_N$, the solution $\lambda^*_N$ cannot be increasing, proving that $\lambda^*_N$ is a non-increasing function of $\textbf{G}$.  Since our value functions ($E\{\max (X^i,\lambda^*_i)\}$) are non-decreasing functions of the thresholds $\lambda^*$s, we have now show that they are non-increasing in $\mathbf{G}$. 
Next assume the non-decreasing property holds for $i = n+1, \cdots, N-1$. Consider $i = n$. We prove this in the cases $\lambda^*_n < c_n$ and $\lambda^*_n \geq c_n$, respectively. For the case $\lambda^*_n \geq c_n$, we have
$
\lambda^*_n \bar{t}^c_N = \int_{\lambda^*_n}^{\bar{X}^n} (x-\lambda^*_n)f_{X^n}(x)dx
$. 
Using similar argument as in the case $i=N$ we know $\lambda^*_n$ is non-increasing in $\textbf{G}$. For the case $\lambda^*_n < c_n$,
$
\lambda^*_n = \frac{\int_{c_n}^{\bar{X}^n}xf_{X^n}(x)dx+c_n\cdot P(X^n \leq c_n)}{1+\bar{t}^c_n}
$, and we get $E\{V_n\} = \int_{c_n}^{\bar{X}^n}xf_{X^n}(x)dx+c_n\cdot P(X^n \leq c_n)$.  Taking the derivative of $E\{V_n\}$ with respect to $\mathbf{G}$ we get
\begin{align}
\partial E\{V_n\}/\partial \textbf{G} &= \frac{\partial [E(X^n)-\int_{0}^{c_n}xf_{X^n}(x)dx+c_nP(X^n \leq c_n)]}{\partial c_n}\cdot \frac{\partial c_n}{\partial \textbf{G}}\\
\partial c_n/\partial \textbf{G} &= \frac{\frac{\partial E\{V_{n+1}\}}{\partial \textbf{G}}(T+t^s_{n+1})-E\{V_{n+1}\}\frac{\partial t^s_{n+1}}{\partial G_{n+1}}}{(T+t^s_{n+1})^2}
\end{align}
By induction hypothesis we know $\frac{\partial E\{V_{n+1}\}}{\partial \textbf{G}} \leq 0$ and $\frac{\partial t^s_{n+1}}{\partial G_{n+1}} \geq 0$. Therefore we conclude $\frac{\partial c_n}{\partial G_n} \leq 0$,$\frac{\partial E\{V_n\}}{\partial \textbf{G}} \leq 0$, completing the induction step and the proof. \qed

\section{Proof of Lemma \ref{staiid}: ergodicity of $\mathbf{G}$}
\label{appendix-ergo}

By Assumption \ref{assupsta} there exists a threshold $\tilde{G}_{i}$ such that
$E\{V_i(G_i)\} < E\{V_{-i} (G_{-i})\}$, $\forall i \in \Omega$, for all $G_i \geq \tilde{G}_i$, where $G_{-i}$ denotes the aggregated load on all other channel except channel $i$, and $E\{V_{-i} (G_{-i})\}$ is defined as the average reward/rate-of-return of all other channels except $i$. 
In this case, the arrivals to all other channels except $i$ will not switch to channel $i$, i.e., under loads $G_i> \tilde{G}_i$ the probability of load $G_i$ drifting higher is 0 almost surely. Define any increasing, unbounded Lyapunov function $L(G_i)$ on $[0,G]$ (e.g., $L(G_i) = \frac{1}{G-G_i}$), we have
$
E_{\tilde{G}_i}[L(\tilde{G}_i)|G_i] \leq L(G_i)
$.
By the Foster-Lyapunov criteria \cite{MeyTwe93} we establish the ergodicity of the system load vector. 
\qed

\section{Proof of Lemma \ref{iidlg}: load balance}
\label{app-lb}
We prove this by induction on $N$. When $N=1$, i.e., the system degenerates to a single channel case, the claim holds obviously. Assume the claim holds for $N=2, \cdots, n-1$, and now consider the case $N=n$.  Suppose we increase the total load from $G$ to $G'$, and assume that without loss of generality the load to channel 1 decreases, i.e., $G_1^{'} < G_1$. By the induction hypothesis, the loads on all other channels have increased, i.e., $G_i^{'}> G_i$, $\forall i\neq 1$.  As a result, their corresponding value functions decrease by the previous lemma, i.e., $E\{V^{'}_i\} < E\{V_i\}$, $\forall i\neq 1$.  This means that the amount switching out of channel 1 must be non-increasing, due to the fact that the threshold of switching $c_1$ is a non-increasing function of $\mathbf G$, 
while the amount switching into channel 1 must be non-decreasing, leading to an overall non-decreasing load on channel 1, which is a contradiction. 
\qed

\section{Proof of Lemma \ref{iidt}: monotonicity of value functions in $T$}
\label{app-C}

When $i=N$, i.e. the last stage, we have
$
\lambda^*_N \bar{t}^c_N = \int_{\lambda^*_N}^{\bar{X}^N} (x-\lambda^*_N)f_{X^N}(x)dx
$. Following a similar argument as in the monotonicity in $G$, with the decrease in $\bar{t}^c_N$, the solution $\lambda^*_N$ cannot be decreasing, proving that $\lambda^*_N$ is a non-decreasing function of $T$. Assume now the claim holds for $i=n+1, \cdots, N-1$. When $i=n$, consider two cases. For the case $\lambda^*_n \geq c_{n}$, we have $\lambda^*_n \bar{t}^c_{n} = \int_{\lambda^*_n}^{\bar{X}^n} (x-\lambda^*_n)f_{X^n}(x)dx$.  We know $\lambda^*_n$ is a non-decreasing function of $T$. For the  case $\lambda^*_n < c_n$,
$
\lambda^*_n = \frac{\int_{c_n}^{\bar{X}^n}xf_{X^n}(x)dx+c_n\cdot P(X^n \leq c_n)}{1+\bar{t}^c_n}
$. We have
$
E\{V_n\} = \int_{c_n}^{\bar{X}^n}xf_{X^n}(x)dx+c_n\cdot P(X^n \leq c_n)
$
and taking the derivative of $E\{V_n\}$ w.r.t. $T$ we have 
\begin{align}
\partial E\{V_n\}/\partial T = \frac{\partial [E(X^n)-\int_{0}^{c_n}xf_{X^n}(x)dx+c_nP(X^n \leq c_n)]}{\partial c_n}\cdot \frac{\partial c_n}{\partial T}
\end{align}
With basic algebra (we will omit here) and combine with the fact $\frac{\partial E\{V_{n+1}\}}{\partial T} \geq 0$ (induction hypothesis) and $\frac{\partial t^s_{n+1}}{\partial T} \geq 0$, we conclude $\frac{\partial c_n}{\partial T} > 0$,$\frac{\partial E\{V_n\}}{\partial T} >0$, completing the induction step and the proof.  \qed


\section{Proof of Lemma \ref{52}: contraction}
\label{lemma52}
For $v,z \in \mathcal F$, we have
\begin{align}
(\mathcal Tv)(x)-(\mathcal Tz)(x) &= \max_{u \in \mathcal U} \{r(u,x)+\eta\cdot \sum_{y \in S}v(y)\cdot \mathcal P^u(y|x)\} \nonumber \\
&- \max_{u \in \mathcal U} \{r(u,x)+\eta\cdot \sum_{y \in S}z(y)\cdot \mathcal P^u(y|x)\}
\end{align}
Let $\mu = \arg\max_{u \in \mathcal U}\{r(u,x)+ \eta\sum_{y \in S}v(y)\cdot \mathcal P^u(y|x)\}$, then
\begin{align}
(\mathcal Tv)(x)&-(\mathcal Tz)(x) =  \{r(\mu,x)+\eta\cdot \sum_{y \in S}v(y)\cdot \mathcal P^{\mu}(y|x)\} - \max_{u \in \mathcal U} \{r(u,x)+\eta\cdot \sum_{y \in S}z(y)\cdot \mathcal P^u(y|x)\} \nonumber \\
&\leq \{r(\mu,x)+\eta\cdot \sum_{y \in S}v(y)\cdot \mathcal P^{\mu}(y|x)\} -  \{r(\mu,x)+\eta\cdot \sum_{y \in S}z(y)\cdot \mathcal P^{\mu}(y|x)\} \nonumber \\
&= \eta \sum_{y \in S}[v(y)-z(y)]\cdot \mathcal P^{\mu}(y|x) \leq \eta \max_{y \in S} |v(y)-z(y)| = \eta ||v-z||
\end{align}
Similarly by reversing the order of $z,v$ we have
$
(\mathcal Tz)(x)-(\mathcal Tv)(x) \leq \eta ||v-z||
$.
Therefore we reach at
$||\mathcal Tv-\mathcal Tz|| \leq \eta ||v-z||$,
i.e., $\mathcal T$ is a contraction. \qed

\section{Proof of Theorem \ref{53} : monotonicity of the value function in $\mathbf G$}
\label{theorem53}
As proved in \cite{Kumar:1986:SSE:40665}, $V_i(x),x \in S,i \in \Omega$ can be interpreted as follows
\begin{align}
V_i (x) = \max_{\mathbf u} E~\sum_{k=0}^{\infty}\beta^{k\cdot t^c_i} \cdot r_i(u_k,x_k)
\end{align}
Consider the expected maximum throughput at the last stage, i.e.
$
V_N(x) = \max_{\mathbf u} E~\sum_{k=0}^{\infty}\beta^{k\cdot t^c_N} \cdot r_N(u_k,x_k).
$
Consider a $G^{'}_N \geq G_N$ which gives us $t^{c'}_N \geq t^c_N$. Consider an arbitrary term in the above sum $\beta^{k\cdot t^{c'}_N}$, and  
there exists a $k'$ such that
$
k'\cdot t^c_N \leq  t^{c'}_N \leq (k'+1)\cdot t^c_N
$.
Together with the fact that $\beta^{t}\cdot \sum_{y}\mathcal P^t(y|x)\cdot y$ is convex w.r.t. $t$ we know
\begin{align}
\max \{\beta^{(k'+1)\cdot t^c_N }\cdot &\sum_{y }\mathcal P^{(k'+1)\cdot t^c_N }(y|x)\cdot y, \beta^{k'\cdot t^c_N }\cdot \sum_{y} \mathcal P^{k'\cdot t^c_N }(y|x)\cdot y\} \nonumber \\
&\geq \beta^{k \cdot t^{c'}_N }\cdot \sum_{y}\mathcal P^{k\cdot t^{c'}_N }(y|x)\cdot y
\end{align}
\begin{align}
V^{'}_N(x) &= \max_{\mathbf u} E~\sum_{k=0}^{\infty}(\beta^{k\cdot t^{c^{'}}_N })^k\cdot r_N(u_k,x_k) \leq \max_{\mathbf u} E~\sum_{k=0}^{\infty}\beta^{k\cdot t^{c}_N } \cdot r_N(u_k,x_k) = V_N(x)
\end{align}
Therefore as $E\{V_N\} = \sum_x \pi_x \cdot V_N(x)$, and we know $E\{V_N\}$ is a non-increasing function of $\mathbf{G}$. This establishes the induction basis. Now assume that the theorem holds for $i = n+1, \cdots, N-1$. Consider the case $i = n$. Assume $\mathbf{G^{'} > G}$. As discussed in the IID section we have
$r^{'}_n(u,x) \leq r_n(u,x)$. (This can be proved by taking derivatives of $c_i$ with respect to $\mathbf{G}$ and by induction hypothesis $\frac{\partial E\{V_{n+1}\}}{\partial \mathbf{G}} \leq 0$). Therefore again similarly as argued above we have 
\begin{align}
V^{'}_n(x) &= \max_{\mathbf u} E~\sum_{k=0}^{\infty}\beta^{k\cdot t^{c'}_n} \cdot r^{'}_n(u_k,x_k) \leq \max_{\mathbf u} E~\sum_{k=0}^{\infty}\beta^{k\cdot t^{c}_n} \cdot r^{'}_n(u_k,x_k) \nonumber \\
&\leq \max_{\mathbf u} E~\sum_{k=0}^{\infty}\beta^{k\cdot t^{c}_n} \cdot r_n(u_k,x_k) = V_n(x)
\end{align}
which completes the induction step. \qed

\section{Backward calculation of the two-dimensional nested stopping policy}
\label{bc}
We describe the process of calculating the threshold for each channel. Note at the last stage of the decision process there is no more channel to switch to; therefore the dynamic program degenerates to a standard rate-of-return problem. The standard optimal stopping rule thus applies and the details are omitted.  By going backward, at a subsequent stage $i<N$, the quantity $E\{V_{i+1}(X^{i+1})\}$ is available, and we have $V_i (x) = \max \{\hat{X}^{i},\frac{T}{t^c_{i}+T}\cdot E \{V_{i}(X^i)\}\}$.  We calculate $c_i$ as $c_i = \frac{T}{T+t^s_{i+1}}E\{V_{i+1}(X^{i+1})\}$, 
and obtain $\frac{\int_{c_i}^{\bar{X}^i}xf_{X^i}(x)dx+c_i\cdot P(X^i \leq c_i)}{1+{t}^c_i/T}$.
If the latter is less than  $c_i$, we are done and take this as the threshold $\lambda^*$. Otherwise, we proceed to a fixed-point equation
$
\lambda = \frac{\int_{\lambda}^{\bar{X}^i} xf_{X^i}(x)dx}{P(X^i \geq \lambda) + {t}^c_i/T}
$ which can be solved iteratively to obtain the threshold. 

\end{document}